\documentclass{sigma}
\newcommand{\braket}[2]{\left< #1 \vphantom{#2} \right|
\left. #2 \vphantom{#1} \right>} % for Dirac bracket
\usepackage{graphicx} % Package pour les images
\newlength{\textlarg} % Pour le texte barre

\begin{document}
\ShortArticleName{Equivalent sets of coherent states of the infinite square well}
\ArticleName{Equivalent sets of coherent states of the 1D infinite square well and properties}
\Author{Marc-Antoine FISET$^\dagger$ and V\'eronique HUSSIN$^\ddagger$}
\AuthorNameForHeading{M.-A. Fiset and V. Hussin}
\Address{$^\dagger$Department of Physics,
McGill University, Montr\'eal, QC H3A 2T8,
Canada}
\Address{$^\ddagger$D\'epartement de math\'ematiques et de statistique,
Universit\'e de Montr\'eal, Montr\'eal, QC H3C 3J7,
Canada}
\Email{marc-antoine.fiset@mail.mcgill.ca, veronique.hussin@umontreal.ca}
%%%\ArticleDates{Received ???, in final form ????; Published online ????}
\Abstract{We prove the equivalence (under some conditions) of two sets of coherent states built for the one-dimensional infinite square well: the so-called generalized and Gaussian Klauder coherent states. We then derive an approximate close expression approaching their probability density and wave function to explore their properties analytically. This process gives thereby explanation of the quasi-classical behaviour of these states in terms of the main observables and the Heisenberg uncertainty product.}

\section{Introduction}
It is well-known \cite{Gazeau-Klauder, Gazeau} that the standard coherent states of the harmonic oscillator (HO) show many attractive properties and that not all of them can be maintained when we consider other quantum systems. A problem is thus to decide the ones that will be taken as pertinent and the ones that are more peripheral when we try to generalize these states to other quantum systems.

The standard coherent states of the HO are constructed as eigenstates of the annihilation operator. For the one-dimensional infinite square well, such eigenstates have also been built and we find a good revue of their properties in the litterature \cite{Dong}. We will call them "generalized coherent states" (GeCS) in the following.

Another type of states has been constructed in order to get a good localization in the phase space of the quantum system under consideration. They are called "Gaussian Klauder coherent states" (GCS) \cite{Fox-Choi2000}. For the HO, they are shown to be a good approximation of the standard coherent states  \cite{Fox-Choi2000}. For the infinite well, they have been analyzed most notably in \cite{Fox-Choi2000}.

In this work, we want to exhibit the relation between the above two constructions in the case of the infinite square well. We also insist on some analytical results on the behaviour of the wave function and the corresponding quantum-classical correspondence.

In Section 2, we define the two sets of coherent states we will be dealing with. Some of their well-known properties are again exhibited. A new result is that we are able to give a relation between the parameters of these states which leads to an equivalence between the two sets. In Section 3, we examine some properties of the GCS specifically. Whereas numerical results have already exposed the main features [4], our approach provides an elegant explanation for them. We manage to approximate the probability density and the GCS by close expressions, from which we will be able to deduce the behaviour of the main observables (the position and the momentum). We will examine as well the minimization of the Heisenberg uncertainty relation to establish the quantum-classical correspondence for these states. Section 4 will be devoted to some conclusions and future works.

\section{Two sets of coherent states of the 1D infinite well}

Let us first set our notational convention concerning the infinite square well  \cite{Fox-Choi2000} to be used throughout this work. A particle of mass $M$ moves in a potential taken to be
\begin{equation}
V(x)=\begin{cases}0,&0<x<L\\\infty,&\text{otherwise}.\end{cases}\nonumber
\end{equation}
The stationary eigenstates and the discrete energies of this system are
\begin{equation}
\psi_n(x)\equiv\sqrt{\frac{2}{L}}\sin{\frac{(n+1)\pi x}{L}},\quad\quad E(n)\equiv\frac{(n+1)^2\pi^2\hbar^2}{2ML^2}=\hbar\omega(n+1)^2,
\label{eigenISW}
\end{equation}
where $n=0,1,2,...$ and $\omega\equiv\frac{\pi^2\hbar}{2ML^2}$.

\subsection{Generalized coherent states}

The GeCS are usually defined as eigenstates of the annihilation operator of the quantum system under consideration. They can be used as long as the Hamiltonian $H$ of the system has a non degenerate spectrum and admits a lowest energy equal to zero. For the infinite well, we thus work with the shifted Hamiltonian $\hbar\omega\mathcal{H}\equiv H-E_0 {\mathbb I}$ instead of $H$. It has the same eigenstates as in (\ref{eigenISW}) but the eigenvalues are now 
\begin{equation}
E(n)-E(0)=\hbar \omega n(n+2)\equiv \hbar \omega \mathcal{E}(n).\nonumber 
\end{equation}

Ladder operators are chosen such that their action on the energy eigenstates is
\begin{equation}
a \psi_n(x)=\sqrt{\mathcal{E}(n)}\psi_{n-1}(x),\quad a^\dagger \psi_n(x)=\sqrt{\mathcal{E}(n+1)}\psi_{n+1}(x).\nonumber 
\end{equation}

Note that other types of GeCS have been constructed by generalizing the preceding action of the ladder operators (see \cite{Dong} for a review). They are usually called Perelomov, Barut-Girardello, Gazeau-Klauder and deformed coherent states. Indeed, we can take
\begin{equation}
A \psi_n(x)=\sqrt{n} f(n)\psi_{n-1}(x),\quad A^\dagger  \psi_n(x)=\sqrt{n+1}f(n+1)\psi_{n+1}(x), \nonumber
\end{equation}
for a positive real function $f(n)$ of the quantum number $n$. Here we limit ourselves to $f(n)=\sqrt{n+2}$ since we have essentially a factorization of the Hamiltonian of the system as for the HO:
\begin{equation}
a^\dagger a \psi_n(x)=\mathcal{E}(n) \psi_n(x)=\mathcal{H} \psi_n(x).\nonumber
\end{equation}

The main difference with respect to the HO case is that the set $\{a, a^\dagger, N\}$, where $N$ is the usual number operator ($N\psi_n(x)\equiv n\psi_n(x)$), satisfies a $su(1,1)$ algebra:
\begin{equation}
[a,N]=a,\quad [a^\dagger,N]=-a^\dagger,\quad [a,a^\dagger]=2\left(N+\frac{3}{2}\right).\nonumber 
\end{equation}

A realization of the ladder operators \cite{Dong} in terms of the position $x$, momentum $p=-i \hbar \frac{d}{dx}$ and number operators is given by ($\alpha\equiv{\frac{\pi}{L}}$):
\begin{equation}
a=\left[\cos(\alpha x)-\frac{i \sin(\alpha x)}{\hbar \alpha} p \frac{1}{N+1}\right] \sqrt{\mathcal{E}(N)},\nonumber 
\end{equation}
\begin{equation}
a^\dagger=\left[\cos(\alpha x)+\frac{i \sin(\alpha x)}{\hbar \alpha} p \frac{1}{N+1}\right] \sqrt{\mathcal{E}(N+1)}.\nonumber 
\end{equation}

The GeCS can be written as a function of the real position $x$, time $t$ and a continuous complex parameter $z$ which is the eigenvalue of the annihilation operator $a$:
\begin{equation}
\Psi_{\text{Ge}}(z;x,t)\equiv\frac{1}{\sqrt{N_\text{Ge}(z)}}\sum_{n=0}^\infty\frac{z^n}{\sqrt{\rho(n)}}e^{-i\omega \mathcal{E}(n)t}\psi_n(x),
\quad\quad \rho(n)=\begin{cases}1,&n=0,\\\prod_{i=1}^n \mathcal{E}(i),&n>0.\end{cases}\label{GeGeneral}
\end{equation}
The normalization factor is
\begin{equation}
N_\text{Ge}(z)\equiv\sum_{n=0}^\infty \frac{|z|^{2n}}{\rho(n)}.\nonumber 
\end{equation}

These states are widely used because they are a direct generalization of the HO coherent states, $\rho(n)$ being essentially the product of the shifted energies. The properties of those states are well-known \cite{Gazeau-Klauder,Gazeau,Dong}. In particular, the resolution of the identity is satisfied as well as time stability and continuity in $z$ and $t$.

We can write \eqref{GeGeneral} more succinctly as \cite{Gazeau}:
\begin{equation}
\Psi_{\text{Ge}}(z;x, t)=\sum_{n=0}^\infty C_n^{\text{Ge}}(z_0,\phi_0)e^{-i\omega n(n+2)t}\psi_n(x),\quad C_n^{\text{Ge}}(z_0,\phi_0)\equiv\frac{z_0^{n+1} e^{-in\phi_0}}{\sqrt{I_2(2 z_0)n!(n+2)!}},\label{Ge}
\end{equation}
using $z=z_0\sqrt{\hbar\omega} e^{-i\phi_0} (z_0,\phi_0\in\mathbb{R})$ and denoting by $I_2$ the second-order modified Bessel function of the first kind.

\subsection{Gaussian Klauder coherent states}

Even though they are less frequently used then the GeCS, the GCS can be built for many different systems as a special superposition of energy eigenstates in order to get a reasonably well localized probability density distribution for a short period of time \cite{Fox-Choi2000}. For real parameters $\phi_0, \ n_0\geq0$ and $\sigma_0>0$, they are defined as
\begin{equation}
\Psi_{\text{G}}(n_0,\sigma_0,\phi_0; x,t)=\sum_{n=0}^\infty C_n^{\text{G}}(n_0,\sigma_0,\phi_0)e^{-i\omega\mathcal{E}(n)t}\psi_n(x),\quad C_n^{\text{G}}(n_0,\sigma_0,\phi_0)=\frac{e^{-\frac{(n-n_0)^2}{4\sigma_0^2}-in\phi_0}}{\sqrt{N_\text{G}(n_0,\sigma_0)}},\label{G}
\end{equation}
where the normalization factor is
\begin{equation}
N_\text{G}(n_0, \sigma_0)=\sum_{n=0}^\infty e^{-\frac{(n-n_0)^2}{2\sigma_0^2}}.\nonumber 
\end{equation}

The resolution of the identity is satisfied as well as time stability and continuity in $n_0$ and $\sigma_0$ \cite{Fox-Choi2000}.

Let us stress the introduction of the same factor $e^{-in\phi_0}$ in  \eqref{Ge} and \eqref{G}. It was not included in the original construction \cite{Fox-Choi2000} but it will help making the connection between the two sets of coherent states. It also has a very simple physical interpretation as we shall show in section 3.2.

\subsection{Equivalence between the two sets of coherent states}

In order to compare our coherent states, we consider $|C_n^{\text{Ge}}(z_0,\phi_0)|^2$ as given from \eqref{Ge} and assume $z_0\gg1$. First we have:
\begin{equation}
|C_n^{\text{Ge}}(z_0,\phi_0)|^2=\frac{z_0^{2n+2}}{I_2(2z_0)n!(n+2)!}=\frac{e^{2z_0}}{I_2(2z_0)}\left[\frac{e^{-z_0}z_0^{n+1}}{(n+1)!}\right]^2\frac{n+1}{n+2}\label{Poisson}.
\end{equation}

The expression inside the bracket on the right-hand side of \eqref{Poisson} is a Poisson distribution in $(n+1)$, that can be approximated by a Gaussian distribution of mean $z_0$ and standard deviation $\sqrt{z_0}$. Looking back at \eqref{G}, we get 
\begin{equation}
|C_n^{\text{Ge}}(z_0,\phi_0)|^2\simeq\frac{e^{2z_0}N_\text{G}(z_0-1,\sqrt{z_0/2})}{2\pi z_0I_2(2z_0)}|C_n^{\text{G}}(z_0-1,\sqrt{z_0/2},\phi_0)|^2\frac{n+1}{n+2}.\label{GeApproxG}
\end{equation}

The leading behaviour of the normalization factor is straightforwardly obtained from a standard Euler-Maclaurin asymptotic expansion:
\begin{align}
\sum_{n=0}^\infty e^{-\frac{(n+1-z_0)^2}{z_0}}&\sim\int_0^\infty e^{-\frac{(n+1-z_0)^2}{z_0}}dn+\frac{1}{2}e^{-\frac{(1-z_0)^2}{z_0}}-\sum_{k=1}^\infty\frac{B_{2k}}{(2k)!}\left.\frac{d^{2k-1}}{dn^{2k-1}}\right|_{n=0} e^{-\frac{(n+1-z_0)^2}{z_0}}\nonumber \\
&=\frac{\sqrt{\pi z_0}}{2}\left[\text{erf}\left(\frac{1-z_0}{\sqrt{z_0}}\right)+1\right]+e^{-\frac{(1-z_0)^2}{z_0}}\left[\frac{1}{2}-\sum_{k=1}^\infty\frac{B_{2k}}{(2k)!}2^{2k-1}+O\left(z_0^{-1}\right)\right],\nonumber
\end{align}
where $B_{2k}$ are Bernoulli numbers. Using the identity
\begin{equation}
\frac{1}{2}\coth{\left(\frac{x}{2}\right)}-\frac{1}{x}=\sum_{k=1}^\infty\frac{B_{2k}}{(2k)!}x^{2k-1},\nonumber 
\end{equation}
valid for $0<|x|<2\pi$, and the asymptotic expansion
\begin{equation}
\text{erf}(x)\sim 1+e^{-x^2}\left[-\frac{1}{\sqrt{\pi}x}+O\left(x^{-3}\right)\right] \quad (x\rightarrow\infty),\nonumber 
\end{equation}
we find
\begin{equation}
N_\text{G}(z_0-1,\sqrt{z_0/2})\sim\sqrt{\pi z_0}+e^{-\frac{(1-z_0)^2}{z_0}}\left[\frac{1}{1-e^2}+O\left(z_0^{-1}\right)\right].\label{Nbehaviour}
\end{equation}

Approximating sums by integral in this way will be a recurring theme in this document. Our analysis will be first order and we will thus typically keep only the dominant behaviour without mentioning the corrections.

Along with the $z_0\rightarrow\infty$ expansion $I_2(2z_0)\sim e^{2z_0}/\sqrt{4\pi z_0}[1+O(z_0^{-1})]$ \cite{Arfken}, \eqref{Nbehaviour} turns \eqref{GeApproxG} into
\begin{equation}
|C_n^{\text{Ge}}(z_0,\phi_0)|^2\simeq[1+O(z_0^{-1})]|C_n^{\text{G}}(z_0-1,\sqrt{z_0/2},\phi_0)|^2\frac{n+1}{n+2}.\nonumber 
\end{equation}

Taking finally into account that only terms with $n$ close to $z_0-1$ contribute significantly (i.e. terms within a few standard deviations from the Gaussian mean), we can approximate by one the $n$-dependent ratio. Matching the phases $\phi_0$ properly, we conclude that the two sets of states are equivalent in the limit $z_0\gg1$ if the parameters are related as $n_0=z_0-1$ and $\sigma_0^2=z_0/2$. We see that there is more freedom in the GCS, where $\sigma_0$ and $n_0$ are a priori independent, than in the GeCS.

\section{Quantum-classical correspondence for the Gaussian Klauder coherent states}

As clear from any introductory Quantum Mechanics textbook, the quantum-classical correspondence of the standard HO coherent states $\Psi_{\text{HO}}(z;x,t)$ relies on two important properties. First, the main observables of these states have classical sinusoidal time dependence. Second, the Heisenberg product saturates the uncertainty relation at any time. This is as close as a quantum state can get to being classical.

These remarkable features all trace to the fact that the probability density can be written exactly
\begin{equation}
|\Psi_{\text{HO}}(z;x,t)|^2=\sqrt{\frac{m\omega}{\pi\hbar}}e^{-\frac{(x-\left\langle x \right\rangle)^2}{2\Delta_x^2}},\quad \Delta_x^2=\left\langle x^2 \right\rangle-\left\langle x \right\rangle^2.\label{HOprobdensity}
\end{equation}
for any time $t$. In this section, we examine how much of these properties survive in the case of the infinite square well GCS. The correspondence between these states and the GeCS given in section 2.3 makes our discussion applicable for both sets of states. We choose to focus on the GCS since their parameters will translate more naturally in terms of the quantum observables.

\subsection{Computed behaviour of the main observables}

Some characteristics of GCS \eqref{G} have been explored by Fox and Choi in \cite{Fox-Choi2000}. They highlighted the fact that the main observables behave quasi-classically for a short period of time before the wave packet decays. In particular, $\braket{\Psi_{\text{G}}}{x |\Psi_{\text{G}}}\equiv\left\langle x \right\rangle$ as a function of time is approximately a triangular wave, which is in good agreement with the classical back and forth motion resulting from bounces on the walls. This behaviour is shown on figure 1a (obtained by summing numerically a finite number of terms from \eqref{G}). Moreover, figure 1b shows that the average momentum $p$ is constant except at regularly spaced bounces, as we shall expect for a classical system \cite{Fox-Choi2000}. 

\begin{center}
\includegraphics[scale=0.5]{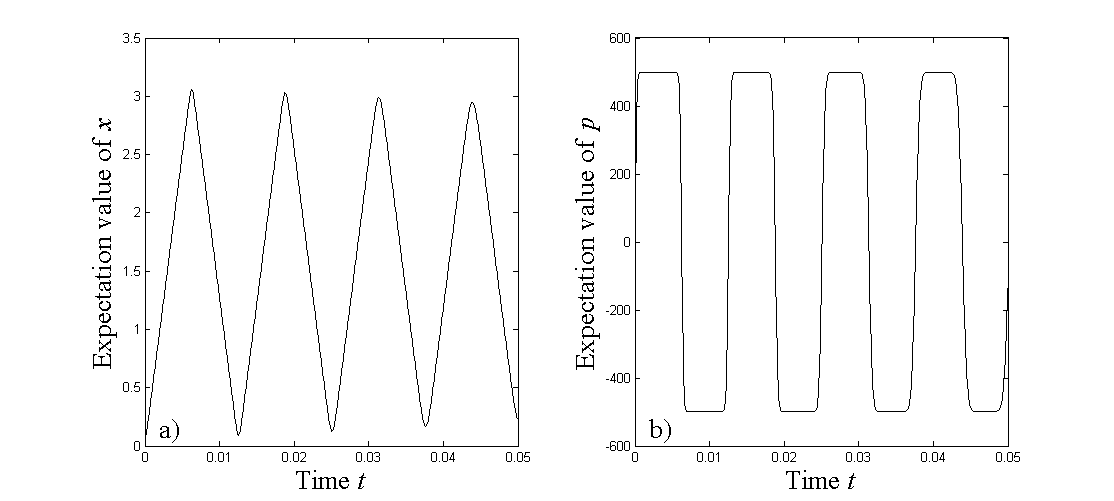}

\textbf{Figure 1} - (a) Expectation value of the position $\left\langle x \right\rangle$ and (b) momentum $\left\langle p \right\rangle$ as a function of time for $n_0=500, \sigma_0=5, \phi_0=\pi/2, L=\pi, \hbar=1$ and $M=1$.
\end{center}

We propose here another way of getting those results with a conceptual improvement. Instead of studying numerically features of $|\Psi_{\text{G}}(n_0,\sigma_0,\phi_0;x,t)|^2$, namely $\left\langle x \right\rangle$ and $\left\langle p \right\rangle$, we focus analytically on the probability density as a whole. We then obtain the behaviour of the main observables as corollaries.

\subsection{Approximate formula for the probability density and properties}

The HO coherent state probability density is a Gaussian wavepacket nicely packaging the position expectation value (see \eqref{HOprobdensity}). The next proposition shows that a similar formula holds for the infinite square well coherent states, as much as permitted by the limited domain $x\in[0,L]$. To our knowledge, it is the first time that this similarity with the HO is pointed out.

\begin{proposition}{The probability density of the infinite square well GCS is such that
\begin{equation}
|\Psi_{\text{G}}(n_0,\sigma_0,\phi_0;x,t)|^2\simeq\frac{1}{\sqrt{2\pi}s}e^{-\frac{(x-X)^2}{2s^2}}\label{probDensity}
\end{equation}
for $x\in[0,L]$, $t>0$ with
\begin{align*}
&X\equiv\frac{\phi_0L}{\pi}+\frac{Pt}{M}, \quad P\equiv\frac{(n_0+1)\pi\hbar}{L},\\
&s\equiv\frac{L}{2\pi\sigma}, \quad \sigma\equiv\sqrt{\frac{\tau}{4\omega(\tau^2+t^2)}}, \quad \tau\equiv(4\omega\sigma_0^2)^{-1}%\\
%&C_\Pi\equiv2{\left[\mathrm{erf}\left(\frac{L-X}{\sqrt{2}s}\right)+\mathrm{erf}\left(\frac{X}{\sqrt{2}s}\right)\right]}^{-1}
\end{align*}
under the conditions $n_0\gg\sigma_0\gg1$, $X\gg s$, $L-X\gg s$ and $t\ll\tau$.}
\end{proposition}
\begin{proof}
See appendix 1.
\end{proof}

As promised, the results of \cite{Fox-Choi2000} can be extracted from this proposition. For example, $X$, which we readily identify with $\left\langle x\right\rangle$, depends linearly on time. In other words, the particle moves with constant momentum $\left\langle p\right\rangle=M\frac{dX}{dt}=P=\frac{(n_0+1)\pi\hbar}{L}$. The period of motion inferred from this momentum and the length of the well, $\frac{2ML^2}{(n_0+1)\pi\hbar}$, also agrees with the one found from Taylor expansions of the quantized energies (see \cite{Gazeau, Aronstein-Strout}).

Because of its linear dependence, $X$ however evades the interval $[0,L]$ after a short time. The approximate formula of proposition 1 is then trivially wrong. This is a minor complication since numerical evidence suggest that, provided we account manually for the discrete bounces, our formula  still approximate correctly the wave packet. We strongly believe that it would be easy to generalize our proposition by having $X$ to be a triangular wave instead of a linear function. However, the regularly spaced moments when the packet is close to the boundaries of the well would still be badly described by an analogue of our close expression.

Unlike the time when the packet hits the wall, $\tau$ has a physical significance worth mentioning. It is an intrinsic property of the GCS that the initially highly localized wave packet decays as time evolves \cite{Fox-Choi2000}. This is reflected in the Lorentzian time-evolution of the width $s$ of the Gaussian packet as defined in the proposition. $\tau$ serves here as a typical order of magnitude for the decay process. The condition $t\ll\tau$ simply establishes our restricted attention to early instants free of this complication.

We finally see from the expression of $X$ that the parameter $\phi_0$ serves as an initial position of the wave packet. This simple interpretation justifies our introduction of that parameter in the definition of \eqref{G}. Another interesting observation we can make is that the maximum of $|\Psi_{\text{G}}(n_0,\sigma, x,t)|^2$ can be found from \eqref{probDensity}. It is close to$\frac{\sqrt{2\pi}\sigma}{L}$ at the middle of the well.

\subsection{Minimization of the uncertainty product}

One of the most important properties of the HO coherent states is the minimization of the Heisenberg uncertainty relation $\Delta\equiv\Delta x\Delta p$ at any time. Since we also noted striking quantum-classical similarities for $t\ll\tau$ for the infinite square well GCS \eqref{G}, we now want to see whether $\Delta$ would be around $\hbar/2$ in some conditions.

\begin{center}
\includegraphics[scale=0.5]{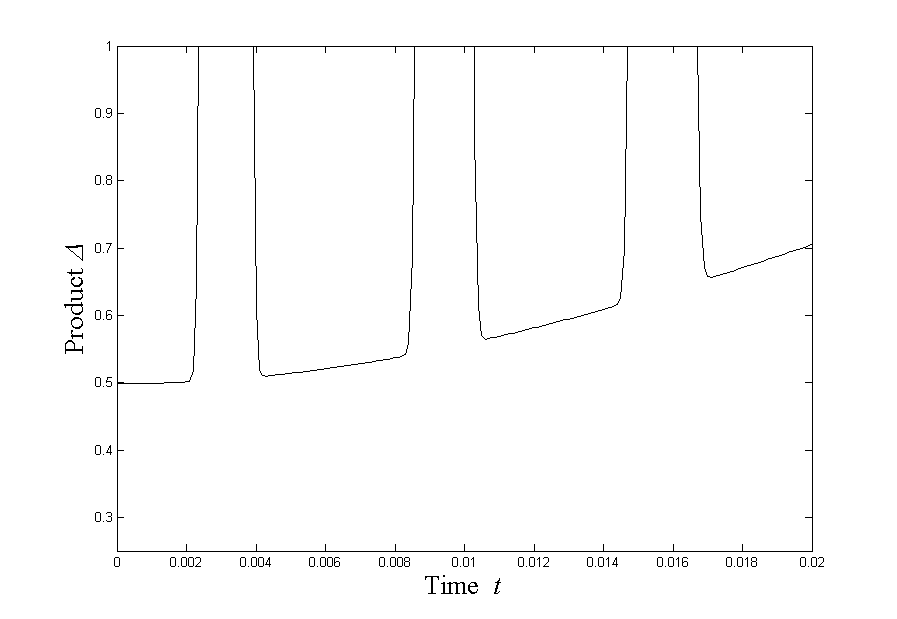}

\textbf{Figure 2} - Uncertainty product $\Delta\equiv\Delta x\Delta p$ for $n_0=50, \sigma_0=5, \phi_0=\pi/2, L=\pi, \hbar=1$ and $M=1$.
\end{center}

Figure 2 shows the numerically calculated uncertainty product for the GCS. It is easy to show from proposition 1 that the peaks of high $\Delta$ coincide with the bouncing of the particle on the walls. These moments set apart, $\Delta$ is close to one half which is the minimal possible value. This can be understood analytically as we will now see. The most general wave function that minimizes $\Delta$ being \cite{Cohen-Tannoudji}
\begin{equation}
\Psi(x,t)=A(t)e^{-\frac{(x-\left\langle x \right\rangle)^2}{4s^2}+\frac{i\left\langle p \right\rangle x}{\hbar}}\label{minimalWaveFunction},
\end{equation}
the corresponding probability density $|\Psi(x,t)|^2$ should certainly show a Gaussian dependence on $x$. Even though we have proven that the GCS have this Gaussian probability density, we know nothing about the wave function. We thus need the following result.

\begin{proposition}
The wave function of the infinite square well GCS satisfies (up to a $x$-independent phase factor)
\begin{equation}
\Psi_{\text{G}}(n_0,\sigma_0,\phi_0;x,t)\simeq\frac{1}{(\sqrt{2\pi}s)^{1/2}}e^{-\frac{(x-X)^2}{4s^2}+\frac{iPx}{\hbar}}\label{waveFunction}
\end{equation}
with the same parameters and under the same conditions as in proposition 1.
\end{proposition}

\begin{proof}
See appendix 2.
\end{proof}

The GCS wave function then has approximately the specific Gaussian form of \eqref{minimalWaveFunction} explaining why the Heisenberg relation reaches its minimum at $t\ll\tau$. The identification of $P$ with $\left\langle p \right\rangle$ we are making here is moreover consistent with the speed of the wave packet found in section 3.2.

Proposition 2 is not really surprising given proposition 1. However, they are complementary. As the proofs show, proposition 1 gives a precise understanding of the decay of the wave packet. In particular, it gives the expression for $\sigma$. On the other hand, proposition 2 gives the extra dependence on $P$ which contributes to explain the minimization of the Heisenberg product. Both consistently exhibit the quasi-classical behaviour of the GCS.

\section{Conclusion}

In this work, we have provided a short review of two well-known coherent states built for the one-dimensional infinite square well: the generalized and Gaussian-Klauder coherent states. We gave a proof that the two sets of states were equivalent for $z_0\gg1$ if the parameters are related as $n_0=z_0-1$ and $\sigma_0^2=z_0/2$.

We then turned to the analysis of the quantum-classical correspondence properties of those states. Using an approximate close expression for the probability density \eqref{probDensity}, we readily obtained the behaviour of the observables of interest. This approach also exhibited clearly the spreading of the wave packet as a function time.

Another close expression for the wave function \eqref{waveFunction} explained the exhibited minimization of the Heisenberg uncertainty product. Both results \eqref{probDensity} and \eqref{waveFunction} were consistent with a Gaussian wave function just as in the case of the harmonic oscillator coherent states.

\section*{Acknowledgements}

This work was supported by the Natural Sciences and Engineering Research Council of Canada (NSERC).

\section*{Appendix 1}

\begin{proof}
From \eqref{G} and the explicit form of $\psi_{n}(x)$ given in \eqref{eigenISW}, the probability density takes the exact form
\begin{equation}
|\Psi_{\text{G}}(x,t)|^2=\frac{1}{N_\text{G}(n_0,\sigma_0)L}\sum_{n'=0}^\infty \sum_{n=0}^\infty e^{-\Phi(n,n';n_0,\sigma_0,\phi_0,t)}\left[\cos{\frac{(n'-n)\pi x}{L}}-\cos{\frac{(n'+n+2)\pi x}{L}}\right].\nonumber
\end{equation}
with
\begin{equation}
\Phi(n,n';n_0,\sigma_0,\phi_0,t)={\frac{(n-n_0)^2+(n'-n_0)^2}{4\sigma_0^2}+i\phi_0(n-n')+i \omega t(n(n+2)-n'(n'+2))}.\nonumber
\end{equation}

The probability density can be separated into two parts: $|\Psi_{\text{G}}(x,t)|^2=P_0 (n_0,\sigma_0,\phi_0;x,t)+P_l (n_0,\sigma_0,\phi_0;x,t)$ where $P_0$ contains the sum with $\cos{\frac{(n'-n)\pi x}{L}}$ and $P_l$ the other sum. Let us start by working out $P_0$ in details. Introducing the new summation index $j=n'-n$ and $u(j)$ that is $0$ for $j\geq0$ and $-j$ for $j<0$, we get
\begin{align}
P_0(n_0,\sigma_0,\phi_0;x,t)&=\frac{1}{N_\text{G}(n_0,\sigma_0)L}\sum_{j=-\infty}^\infty \sum_{n=u(j)}^\infty e^{-\Phi(n,n+j;n_0,\sigma_0,\phi_0,t)}\cos{\frac{j\pi x}{L}}\nonumber\\
&=\frac{1}{N_\text{G}(n_0,\sigma_0)L}\sum_{j=-\infty}^\infty e^{-\frac{j^2}{4\sigma_0^2}+ij^2\omega t+ij\phi_0}\cos{\frac{j\pi x}{L}}\sum_{n=u(j)}^\infty e^{-\frac{(n-n_0)^2+j(n-n_0)}{2\sigma_0^2}+2i\omega tj(n+1)}.\nonumber
\end{align}

If $n_0\gg1$, it makes no difference to use minus infinity in place of $u(j)$ since $e^{-\Phi(n,n+j;n_0,\sigma_0,\phi_0,t)}$ only selects terms near $j=0$ and $n=n_0$ which is far from $u(j)$. The approximation of the second sum by an integral yields the dominant behaviour
\begin{equation}
\sum_{n=-\infty}^\infty e^{-\frac{(n-n_0)^2+j(n-n_0)}{2\sigma_0^2}+2i\omega tj(n+1)}\simeq\sqrt{2\pi}\sigma_0 \ e^{\frac{j^2}{8\sigma_0^2}-2j^2\omega^2t^2\sigma_0^2-ij^2\omega t+2ij\omega t(n_0+1)}.\nonumber
\end{equation}

Hence, recycling $N_\text{G}(n_0,\sigma_0)\simeq\sqrt{2\pi}\sigma_0$ from \eqref{Nbehaviour}, we get
\begin{align}
P_0(n_0,\sigma_0,\phi_0;x,t)&\simeq\frac{1}{L}\sum_{j=-\infty}^\infty \cos{\frac{j\pi x}{L}}e^{-\frac{j^2}{8\sigma_0^2}-2j^2\omega^2\sigma_0^2t^2+2i\omega tj(n_0+1)+ij\phi_0} \nonumber\\
&=\frac{1}{L}+\frac{2}{L}\sum_{j=1}^\infty e^{-\frac{j^2}{8\sigma^2}}\cos{\frac{j\pi x}{L}}\cos{[j(\phi_0+2\omega t(n_0+1))]}.\label{P_0}
\end{align}

Note the convenient introduction of $\sigma$ as defined in the proposition. A very similar derivation for $P_l(n_0,\sigma,\phi_0;x,t)$ gives
\begin{equation}
P_l(n_0,\sigma_0,\phi_0;x,t)\simeq-\frac{e^{-2\sigma^2(\phi_0+2\omega t(n_0+1))^2}}{L}\sum_{j=0}^\infty e^{-\frac{(j-2n_0)^2}{8\sigma^2}}\cos{\frac{j\pi x}{L}}.\label{P_lSeries}
\end{equation}

Now the question is how to interpret the Fourier series \eqref{P_0} and \eqref{P_lSeries}. Let us write in such a way the even $2L$-periodic extension of the Gaussian function
\begin{equation}
\Pi(X,s,\gamma;x)\equiv\frac{1}{\sqrt{2\pi}s}e^{-\frac{(x-X)^2}{2s^2}}\cos{\gamma x}=\frac{a_0}{2}+\sum_{j=1}^\infty a_j\cos{\frac{j\pi x}{L}}\quad\quad x,x_0\in[0,L],\label{Pi}
\end{equation}
with $\gamma=0$. The $a_j$ are given by
\begin{equation}
a_j=\frac{1}{\sqrt{2\pi}sL}\int_0^L f\left(X,s,2,\frac{j\pi}{L};x\right)
+f\left(X,s,2,-\frac{j\pi}{L};x\right)dx,\label{a_j}
\end{equation}
where
\begin{equation}
f\left(X,s,\alpha,\beta;x\right)\equiv e^{-\frac{(x-X)^2}{\alpha s^2}+i\beta x}.\label{f}
\end{equation}

The integration of \eqref{f} can be carried out explicitly, but the result is simpler if
\begin{equation}
s^2|\beta|\ll\frac{2}{|\alpha|}X\quad\text{and}\quad s^2|\beta|\ll\frac{2}{|\alpha|}(L-X),\label{approxSimplificationOnF}
\end{equation}
in which case
\begin{equation}
\int_0^L f\left(X,s,\alpha,\beta;x\right)dx\simeq\frac{\sqrt{\pi\alpha}s}{2}e^{i\beta X-\frac{\alpha\beta^2s^2}{4}}\left[\text{erf}\left(\frac{L-X}{\sqrt{\alpha}s}\right)+\text{erf}\left(\frac{X}{\sqrt{\alpha}s}\right)\right].\label{integralF}
\end{equation}

\eqref{a_j} and \eqref{integralF} then yield
\begin{equation}
a_j\simeq\frac{2}{L}e^{-\frac{j^2}{8\sigma^2}}\cos{\frac{j\pi X}{L}}\quad\text{if }X\gg s\text{ and }L-X\gg s.\nonumber
\end{equation}

Comparing with \eqref{P_0}, this means $P_0(n_0,\sigma,\phi_0;x,t)\simeq \Pi(X,s,0;x)$ provided \eqref{approxSimplificationOnF} holds. A quick validation convinces that this is the case for the conditions given in the proposition. A similar development gives
\begin{equation}
P_l(n_0,\sigma_0,\phi_0;x,t)\simeq-e^{-\frac{X^2}{2s^2}}\Pi(X,s,\frac{2\pi n_0}{L};x).\nonumber
\end{equation}
but we realize that $P_l$ is actually negligible in front of $P_0$. We simply drop it to complete the proof.
\end{proof}

Some comments are in order concerning $P_l$. It is interesting to notice that it introduces the signature of fine oscillations of period $\frac{L}{n_0}$. These oscillations were already observed in \cite{Fox-Choi2000} when the wave packet is near the boundaries of the well. They allow the wave packet to get nearer of the walls than it would without deforming in that way. We see here that they arise because of $P_l$. The latter acts as a border correction on $P_0$, which embodies the dominant resemblance with a Gaussian probability density.
Strangely, the proof did not yield the border contribution
\begin{equation}
P_r(n_0,\sigma_0,\phi_0;x,t)\equiv-e^{-\frac{(L-X)^2}{2s^2}}\Pi(X,s,\frac{2\pi n_0}{L};x),\nonumber
\end{equation}
which we expect based on the obvious requirement that the solution must behave symmetrically about the middle of the well.

\section*{Appendix 2}

\begin{proof}
Let us Fourier-expand the odd $2L$-periodic extension of the Gaussian function
\begin{equation}
(\sqrt{2\pi}s)^{-1/2}f(X,s,4,\frac{P}{\hbar};x)=\sum_{n=1}^\infty b_n\sin{\frac{n\pi x}{L}}\label{FourierOdd}
\end{equation}

with $f(X,s,\alpha,\beta;x)$ as defined by \eqref{f}. The coefficient $b_n$ is given by
\begin{equation}
b_n=\frac{(\sqrt{2\pi}s)^{-1/2}}{iL}\int_0^L f(X,s,4,\frac{P}{\hbar}+\frac{n\pi x}{L};x)-f(X,s,4,\frac{P}{\hbar}-\frac{n\pi x}{L};x) dx.\nonumber
\end{equation}

The first integral is negligible for $n_0\gg \sigma_0$ and $t\ll\tau$. Refering to \eqref{integralF}, this expression becomes
\begin{equation}
b_n\simeq\frac{i(\sqrt{8\pi}s)^{1/2}}{L}e^{i\left(\frac{P}{\hbar}-\frac{n\pi}{L}\right)X-s^2\left(\frac{P}{\hbar}-\frac{n\pi}{L}\right)^2}.\nonumber
\end{equation}

Up to re-indexing the sum and dropping irrelevant phase factors, \eqref{FourierOdd} becomes
\begin{equation}
(\sqrt{2\pi}s)^{-1/2}f(X,s,4,\frac{P}{\hbar};x)=\frac{1}{(\sqrt{2\pi}\sigma)^{1/2}}\sum_{n=0}^\infty e^{-in\phi_0-i\omega t(n+1)(n_0+1)-\frac{1}{4\sigma^2}(n-n_0)^2}\psi_n(x).
\end{equation}

The similarity with \eqref{G} is now clearer. We complete the proof by using again the fact that relevant $n$'s are close to $n_0$, which leads to the identification $\hbar\omega(n+1)(n_0+1)\simeq E(n)$. The expansion \eqref{Nbehaviour} finally lead to $N_{G}(n_0,\sigma_0)\simeq\sqrt{2\pi}\sigma$ for $t\ll\tau$, which completes the proof.
\end{proof}

\LastPageEnding

\end{document}